\documentclass[a4paper,UKenglish,cleveref, autoref]{lipics-v2019}

\usepackage{color}
\usepackage{graphicx}
\usepackage{amssymb,amsmath,amsthm, amsfonts}
\usepackage{txfonts}
\newtheorem{obs}{Observation}
\newtheorem{reducao}{Reduction Rule}

\newtheorem{branch}{Branching}

\usepackage[boxruled,longend,linesnumbered,norelsize]{algorithm2e}
\newcommand{\RkSB}[1]{\textcolor{blue}{SB: #1}}

\title{Width Parameterizations for {Knot-free~Vertex~Deletion} on Digraphs} 

\titlerunning{KFVD}

\author{Stéphane Bessy}{Université de Montpellier - CNRS, LIRMM, Montpellier, France }{stephane.bessy@lirmm.fr}{}{}

\author{Marin Bougeret}{Université de Montpellier - CNRS, LIRMM, Montpellier, France  }{marin.bougeret@lirmm.fr}{}{}

\author{Alan D. A. Carneiro}{Universidade Federal Fluminense - Instituto de Computação, Niterói, Brazil }{aaurelio@ic.uff.br}{}{}

\author{Fábio Protti}{Universidade Federal Fluminense - Instituto de Computação, Niterói, Brazil}{fabio@ic.uff.br}{}{}

\author{Uéverton S. Souza\footnote{corresponding author}}{Universidade Federal Fluminense  - Instituto de Computação, Niterói, Brazil }{ueverton@ic.uff.br}{}{}

\authorrunning{Bessy et al.}

\Copyright{Stéphane Bessy, Marin Bougeret, Alan D.A. Carneiro, Fábio Protti, Uéverton S. Souza}

\ccsdesc[100]{Design and analysis of algorithms~Parameterized complexity and exact algorithms}
\ccsdesc[100]{Mathematics of computing~Graph theory}

\keywords{Knot, deadlock, width measure, FPT, W[1]-hard, directed feedback vertex set}

\category{}

\relatedversion{}

\supplement{}

\funding{Supported by Grant E-26/203.272/2017, Rio de Janeiro Research Foundation (FAPERJ) and by Grant 303726/2017-2, National Council for Scientific and Technological Development (CNPq).}

\acknowledgements{We thank Ignasi Sau for introducing Alan Carneiro to Stéphane Bessy and Marin Bougeret.}

\nolinenumbers 

\hideLIPIcs  

\newcommand{\fixme}[1]{\textcolor{red} {#1}}

\begin{document}

\maketitle

\begin{abstract}
A knot in a directed graph $G$ is a strongly connected subgraph $Q$ of
$G$ with at least two vertices, such that no vertex in $V(Q)$ is an
in-neighbor of a vertex in $V(G)\setminus V(Q)$. Knots are important
graph structures, because they characterize the existence of deadlocks
in a classical distributed computation model, the so-called
OR-model. Deadlock detection is correlated with the recognition of
knot-free graphs as well as deadlock resolution is closely related to
the {\sc Knot-Free Vertex Deletion (KFVD)} problem, which consists of
determining whether an input graph $G$ has a subset $S \subseteq
V(G)$ of size at most $k$ such that $G[V\setminus S]$ contains no
knot. Because of natural applications in deadlock resolution, {\sc KFVD}
is closely related to {\sc Directed Feedback Vertex Set}.
In this paper we focus on graph width measure parameterizations for
{\sc KFVD}. First, we show that: (i) {\sc KFVD} parameterized by the
size of the solution $k$ is W[1]-hard even when $p$, the length of a longest directed path of the input graph, as well as $\kappa$, its Kenny-width, are bounded by constants, and we remark that {\sc KFVD} is para-NP-hard even considering many directed width measures as parameters, but in FPT when parameterized by clique-width; (ii) {\sc KFVD} can be solved in time $2^{O(tw)}\times n$, but assuming ETH it cannot be solved in $2^{o(tw)}\times n^{O(1)}$, where $tw$ is the treewidth of the underlying undirected graph.
Finally, since the size of a minimum directed feedback vertex set ($dfv$) is an upper bound for the size of a minimum knot-free vertex deletion set, we investigate parameterization by $dfv$ and we show that (iii) {\sc KFVD} can be solved in FPT-time
parameterized by either $dfv+\kappa$ or $dfv+p$; and it admits a Turing kernel by the distance to a DAG having an Hamiltonian path (another parameter larger than $dfv$). Results of $(iii)$ cannot be improved when replacing $dfv$ by $k$ due to $(i)$.
\end{abstract}
  
\section{Introduction}

The study of the {\sc Knot-Free Vertex Deletion} problem emerges from
its application in resolution of deadlocks, where a deadlock is
detected in a distributed system and then a minimum cost
deadlock-breaking set must be found and removed from the system. More
precisely, distributed computations are usually represented by
directed graphs called {\em wait-for graphs}. In a wait-for graph $G =
(V,E)$, the vertex set $V$ represents processes, and the set $E$ of
directed arcs represents wait conditions~\cite{barbosa1998}. An arc
exists in $E$ directed away from $v_i\in V$ towards $v_j\in V$ if
$v_i$ is blocked waiting for a signal from $v_j$. The graph $G$
changes dynamically according to a set of prescribed rules (the {\em
  deadlock model}), as the computation progresses. In essence, the
deadlock model governs how processes should behave throughout
computation, i.e., the deadlock model specifies rules for vertices
that are not \textit{sinks} (vertices with at least one out-neighbor) in $G$ to become
sinks~\cite{barbosa2002} (vertices without out-neighbors). The two main classic deadlock models are the
{\sc AND model}, in which a process $v_i$ can only become a sink when
it receives a signal from {\em all} the processes in $N^+(v_i)$, where
$N^+(v_i)$ stands for the set of out-neighbors of $v_i$ (a conjunction of
resources is needed); and the {\sc OR model}, in which it suffices for
a process $v_i$ to become a sink to receive a signal from {\em at
  least one} of the processes in $N^+(v_i)$ (a disjunction of resources is
sufficient).
Distributed computations are dynamic, however deadlock is a stable
property, in the sense that once it occurs in a consistent global
state of a distributed computation, it still holds for all the
subsequent states. Therefore, as is typical in deadlock studies,
$G$ represents a static wait-for
graph that corresponds to a {\it snapshot} of the distributed
computation in the usual sense of a consistent global
state~\cite{chandy1985}.
Thus, the motivation of our work comes from {\em deadlock resolution},
where deadlocks are detected into a consistent global state $G$, and must be solved through some external intervention such as
aborting one or more processes to break the circular wait condition
causing the deadlock.

Deadlock resolution problems differ according to the considered
deadlock model, i.e., according to the graph structure that
characterizes the deadlock situation.  In the {\sc AND-model}, the
occurrence of deadlocks is characterized by the existence of
cycles~\cite{barbosa2002,barbosa-et-al-16}. Therefore, deadlock
resolution by vertex deletion in the {\sc AND-model} corresponds precisely to the well-known
{\sc Directed Feedback Vertex Set (DFVS)} problem, proved to be NP-hard in the
seminal paper of Karp~\cite{karp1972}, and proved to be FPT
in~\cite{chen2008fixed}. On the other hand, the occurrence of
deadlocks in wait-for graphs $G$ working according to the OR-model are
characterized by the existence of {\em knots} in
$G$~\cite{barbosa-et-al-16,holt1972}. A knot in a directed graph $G$
is a strongly connected subgraph $Q$ of $G$ with at least two vertices
such that there is no arc $uv$ of $G$ with $u \in V(Q)$ and $v \notin
V(Q)$. Thus, deadlock resolution by vertex deletion in the {\sc
  OR-model} can be viewed as the following problem.
\smallskip

\noindent    \fbox{
        \parbox{.965\textwidth}{
\noindent
{\sc \textsc{Knot-Free Vertex Deletion (KFVD)}}

\noindent
\textbf{Instance}: A directed graph $G=(V,E)$; a positive integer $k$.


\noindent
\textbf{Question}: Determine if $G$ has a set $S \subset V(G)$ such that $|S| \leq k$ and $G[V\setminus S]$ is knot-free.
}
}
\smallskip

Notice that a digraph $G$ is knot-free if and only if for any vertex $v$ of $G$, $v$ has a path to a sink.

In~\cite{carneiro2019deadlock}, Carneiro, Souza, and Protti proved
that {\sc KFVD} is NP-complete; and, in~\cite{cocoon2018}, it was shown that {\sc KFVD} is W[1]-hard when parameterized by $k$.

 {\sc KFVD} is closely related to {\sc DFVS} not only because of their
 relation to deadlocks, but also some structural similarities between
 them: the goal of {\sc DFVS} is to obtain a direct acyclic graph
 (DAG) via vertex deletion (in such graphs \emph{all} maximal directed
 paths end at a sink); the goal of {\sc KFVD} is to obtain a knot-free
 graph, and in such graphs for every vertex $v$ there \emph{exists} at
 least one maximal path containing $v$ that ends at a sink. Finally,
 every directed feedback vertex set is a knot-free vertex deletion
 set; thus an optimum for {\sc DFVS} provides an upper bound for {\sc
   KFVD}.  Although {\sc Directed Feedback Vertex Set} is a well-known
 problem, this is not the case of {\sc Knot-Free Vertex Deletion},
 which we propose to analyze more deeply in this work.

Let $S$ be a solution for {\sc KFVD}, and let $Z$ be the set of sinks
in $G[V\setminus S]$. One can see that any $v \in V \setminus S$ has a
path (that does not use any vertex in $S$) to a vertex in $Z$. Thus,
{\sc KFVD} can be seen as the problem of creating a set $Z$ of sinks
(doing at most $k$ vertex removals) such that every remaining vertex
has a path (in $G[V \setminus S]$) to a vertex in $Z$. In this paper,
we denote the set of deleted vertices by $S$, and the set of sinks in
$G[V \setminus S]$ by $Z$.

To get intuition on {\sc KFVD}, note that the choice of the vertices
to be removed must be carefully done, since the removal of a subset of
vertices can turn some strongly connected components into new knots
that will need to be broken by the removal of some internal
vertices. Ideally, it is desirable to solve the current knots by
removing as few vertices as possible for each knot, without creating
new ones. Unfortunately, the generation of other knots can not
always be avoided.



In~\cite{cocoon2017,carneiro2019deadlock}, Carneiro, Souza, and Protti
present a polynomial-time algorithm for {\sc KFVD} in graphs with
maximum degree three. They also show that the problem is NP-complete
even restricted to planar bipartite graphs $G$ with maximum degree four. Later,
in~\cite{cocoon2018}, a parameterized analysis of {\sc KFVD} is
presented, where it was shown that: {\sc KFVD} is W[1]-hard when
parameterized by the size of the solution; and it can be solved in
$2^{k\log \varphi}n^{O(1)}$ time, but assuming SETH it cannot be
solved in $(2-\epsilon)^{k\log \varphi}n^{O(1)}$ time, where $\varphi$
is the size of the largest strongly connected subgraph.


Since the introduction of directed treewidth, much effort has been
devoted to identify algorithmically useful digraph width
measures~\cite{Oliveira2016}. Useful width measures imply
polynomial time tractability for many combinatorial problems on
digraphs of constant width. Since {\sc KFVD} is W[1]-hard when
parameterized by $k$, in this paper we investigate the ecology
of width measures in order to find useful parameters to solve {\sc KFVD} in
FPT time. First, taking $k$ as parameter, we show that {\sc KFVD}
remains W[1]-hard even on instances with both longest directed path
and K-width bounded by constants. From the same reduction, it
follows that {\sc KFVD} is para-NP-hard even considering many width
measures as parameters, such as directed treewidth and DAG-width.
Contrasting with the hardness of {\sc KFVD} on several directed width measure
parameterizations, we show that {\sc KFVD} is FPT when parameterized by
the clique-width of the underlying undirected graph; and it can be solved in $2^{O(tw)}\times n$ time,
but assuming ETH it cannot be solved in $2^{o(tw)}\times n^{O(1)}$ time,
where $tw$ is the treewidth of the underlying undirected graph.
After that, we consider the most natural width parameter related to {\sc KFVD},
the size of a minimum directed feedback vertex set ($dfv$). Such a
parameter is at the same time a measure of the distance from the input
graph to a DAG as well as an upper bound for the size of a
minimum knot-free vertex deletion set.  We show that {\sc KFVD} can be
solved in FPT time either parameterized by $dfv$ and K-width, or $dfv$ and the length of a
longest directed path.  The complexity of {\sc  KFVD} parameterized only by $dfv$ remains open. Finally, we present a polynomial Turing kernel when we are given a special directed feedback vertex set whose removal gives a DAG having a directed Hamiltonian path.

In the rest of this section we give necessary definitions and concepts used in this work. In Section~\ref{sec:redrules} we present some useful observations and preliminary results. In Section~\ref{sec:digmeasures} we discuss digraph width measures and show the W[1]-hardness. In Section~\ref{sec:tw} we discuss the consequences of treewidth parameterization. In Section~\ref{sec:feedback} we explore the directed
feedback vertex set number as a parameter. Finally, Section~\ref{sec:hamiltonian} considers the parameterization by distance to a DAG having a Hamiltonian path.
%


\noindent \textbf{Additional notation.}  We use standard
graph-theoretic and parameterized complexity notations and concepts,
and any undefined notation can be found
in~\cite{bondy1976,cygan2015parameterized}.
%
We consider here directed graphs. Given a vertex $v$ and a subset of vertices $Z$, we say that there is a path from $v$ to $Z$ iff there exists $z \in Z$ such that there is a
$vz$-(directed) path.  For $v \in V(G)$, let $D(v)$ denote the set of
descendants of $v$ in $G$ , i.e. nodes that are reachable from $v$
by a non-empty directed path. Given a set of vertices $C = \{v_1,v_2, \dots, v_p\}$ of $G$, we define $D(C)=\bigcup_{i=1}^{p} D(v_i)$.
Let $A(v_i)$ denote the set of ancestors of $v_i$ in $G$, i.e., nodes
that reach $v_i$ through a non-empty directed path. We also define
$A[v_i] = A(v_i) \cup \{v_i\}$, and given a set of vertices $C = \{v_1,
v_2, \dots, v_p\}$ of $G$, we define $A(C)=\bigcup_{i=1}^{p}
A(v_i)$. For a vertex $v$ of $G$, the out-neighborhood of $v$ is
denoted by $N^+(v)=\{u | vu \in E\}$, and given a set of vertices $C =
\{v_1, v_2, \dots, v_p\}$, we define $N^+(C)= \bigcup_{i=1}^p N^+(v_i)
\setminus C$.
We refer to a  Strongly Connected Component as an SCC. 
A knot in a directed graph $G$ is an SCC $Q$ of $G$ with at least two
vertices such that there is no arc $uv$ of $G$ with $u \in V(Q)$ and
$v \notin V(Q)$.  Finally, a sink (resp. a source) of $G$ is a vertex
with out-degree $0$ (resp. in-degree $0$).  Given a subset of vertices
$S$, we denote $G_S = G[S]$ and $\bar{S}=V \setminus S$. Thus,
$G_{\bar{S}}$ denote the graph obtained by removing $S$.

We denote by $dfv(G)$ the size of a minimum directed feedback vertex
set of $G$.  We generally use $F$ to denote a directed feedback vertex
set and by $R$ the remaining subset, i.e., $R=V \setminus F$.  The
length of a longest directed path of $G$ is denoted by $p(G)$.  The
Kenny-width~\cite{ganian2014} or K-width of $G$ is denoted by
$\kappa(G)$ and is the maximum number of distinct directed $st$-paths
in $G$ over all pairs of distinct vertices $s, t \in V(G)$, where two
$st$-paths are distinct iff they do not use the exact same set of
arcs.
For any function $g$ (like $dfv$, $\kappa$, $p$), $g(G)$ will be
denoted simply by $g$ when the considered graph $G$ can be deduced from the context.  In
what follows we denote by {\sc $g$-KFVD} the KFVD problem parameterized by
$g$ ($g = k$ denotes the parameterization by the solution size).

\if 10 

\medskip

\noindent \textbf{Parameterized Complexity.}  A parameterized problem
$\Pi$ is called Fixed Parameter Treatable (FPT) if there is an
algorithm $A$ (called FPT-algorithm) that computes every instance $I =
(\chi, k)$ correctly and decides if $I$ is a {\em yes}- or {\em no}-
instance in time $f(k).|\chi|^{O(1)}$ for some computable function
$f$. Thus, if $k$ is set to a small value, the growth of the function
in relation to $\chi$ is relatively small.

\RkSB{I think that the following paragraph is too detailled to define
  W[1], to shorten if you agree...} Within parameterized problems the
$W$ hierarchy ($FPT \subseteq W[1] \subseteq W[2] \subseteq \dots
\subseteq W[t] \subseteq W[P]$) is a collection of computational
complexity classes that accounts the level of parameterized
intractability. The weft of a circuit is the maximum number of large
nodes on a path from an input node to the output node. We denote by
$\mathcal{C}_{t,d}$ the class of circuits with weft at most $t$ and
depth at most $d$. For $t\geq 1$, a parameterized problem $\Pi$
belongs to the class $W[t]$, if there is a FPT-reduction from $\Pi$ to
{\sc Weighted Circuit Satisfiability} on $\mathcal{C}_{t,d}$, for some
$d\geq 1$.

The transformation called {\em FPT-reduction} is presented next:

\begin{definition}
\textbf{\cite{flum2006} FPT-reduction:} Let $\Pi(k)$ and $\Pi'(k')$ be
two parameterized problems where $k' \leq f(k)$ for some computable
function $f:\mathbb{N} \rightarrow \mathbb{N}$. A FPT-reduction of
$\Pi(k)$ to $\Pi'(k')$ is a reduction $R$ such that:
\begin{enumerate}
\item[i)] for all $x$, we have $x \in \Pi(k)$ if and only if $R(x) \in \Pi'(k')$;
\item[ii)] $R$ is computable in FPT time (in relation to $k$).
\end{enumerate}
\end{definition}

\RkSB{Maybe it could be good to say that pi FPT implies pi' FPT...}

\RkSB{$PH$ and $Sigma_p^3$ are unused..}\\

\noindent {\bf Additional notation.}  We use $PH$ to denote the
polynomial hierarchy, and $\Sigma_p^3$ to denote its third level.
\fi

\section{Preliminaries}
\label{sec:redrules}

In this section we present some useful remarks and reduction rules.
Remind that in the decision version of the problem we are given $G$
and a positive integer $k$.



\noindent
The first observation is immediate, as if we can make the graph
acyclic, then it will be knot-free.

\begin{obs}\label{viability}
If $k\ge dfv(G)$ then $G$ is a {\em yes}-instance.
\end{obs}

\noindent The two others observations are less obvious but rather natural.

\begin{obs}\label{lem:re-add}
  Let $S$ be a solution with set of sinks $Z$ in $G_{\bar{S}}$, and $s
  \in S$.  Let $S' = S \setminus \{s\}$ and $Z'$ be the set of sinks
  of $G_{\bar{S'}}$.  If there is a path from $s$ to $Z'$ in
  $G_{\bar{S'}}$ then $S'$ is also a solution.
\end{obs}

\begin{proof}
 Let $u \in V(G_{\bar{S'}})$. Let us prove that $u$ has a path to $Z'$ in $G_{\bar{S'}}$.
 If $u=s$ then it is clear by assumption. Suppose now that $u \neq s$.
 As $S$ is a solution, let $P$ be a $uz$-path in $G_{\bar{S}}$ from $u$ to a sink $z \in Z$.
As $V(G_{\bar{S}}) \subseteq V(G_{\bar{S'}})$, $P$ still exists in $G_{\bar{S'}}$. Thus, if $z \in Z'$ we are done. Otherwise, it implies that there is $s \in N^+(z)$ such that $P'=(u,\dots,z,s)$ is a $us$-path in $G_{\bar{S'}}$. As $s$ has a path to $Z'$ in $G_{\bar{S'}}$, we obtain the desired result.
\end{proof}

\noindent  Informally, after deleting a vertex $s$, we can add $s$ back to
 the graph when it is certain that $s$ has a path to a sink in the
 current graph. This is detailed by the following lemma and its corollary.

\begin{lemma}\label{lem:addback}
Let $S$ be a solution with set of sinks $Z$ in $G_{\bar{S}}$.  If
there exists $s \in S$ with $s \notin N^+(Z)$, then $S' = S \setminus \{s\}$ is also a solution.
\end{lemma}

\begin{proof}
Let $Z'$ be the set of sinks of $G_{\bar{S'}}$.
According to Observation~\ref{lem:re-add}, it suffices to prove that there is a path from $s$ to $Z'$ in $G_{\bar{S'}}$. If $s$ is a sink in $G_{\bar{S'}}$ we are done. Otherwise, there exists an arc $su$ in $G_{\bar{S'}}$, with $u \in V(G_{S})$. As $S$ is a solution, either $u$ is a sink and we are done, or,  there exists a $uz$-path $P$ in $G_{\bar{S}}$ with $z \in Z$. As $V(G_{\bar{S}}) \subseteq V(G_{\bar{S'}})$, $P$ still exists in $G_{\bar{S'}}$, 
and $s \notin N^+(Z)$, $z$ is still a sink in $G_{\bar{S'}}$.
\end{proof}

The following corollary is immediate.

\begin{corollary}\label{lem:NZ}
  In any optimal solution $S$ with set of sinks $Z$ in $G_{\bar{S}}$,
  we have $N^+(Z)=S$.
\end{corollary}

\begin{obs}\label{obs:degree}
Let $S$ be a knot-free vertex deletion with set of sinks $Z$ in
$G_{\bar{S}}$. If $|S|\leq k$ then for any vertex $v$ with $d^+(v) >
k$ it holds that $v\notin Z$.
\end{obs}

\noindent
To complete the previous observations, we can design two general
reduction rules.

\begin{reducao}\label{sccsize1}
If $v \in V(G)$ is an SCC of size one then remove $A[v]$.
\end{reducao}


\begin{proof}
Let $G'$ be the graph obtained by removing $A[v]$.
Let of first show that $(G,k)$ is a {\em yes}-instance implies that $(G',k)$ is also a {\em yes}-instance.
Let $S$ be a solution of $G$ of size at most $k$ with set of sinks $Z$ in $G_{\bar{S}}$.
Let $S' = S \setminus A[v]$, and $Z'$ the set of sinks in $G'_{\bar{S'}}$. Let us prove that every $u \in V(G'_{\bar{S'}})$ has a path ot $Z'$ in $G'_{\bar{S'}}$. Let $u \in V(G'_{\bar{S'}})$. 
As $u$ is also in $V(G_{\bar{S}})$, there is a $uz$-path $P$ in $G_{\bar{S}}$ where $z \in Z$. As $u \notin A[v]$, $V(P) \cap A[v] = \emptyset$ and thus, the path $P$ still exists in $G'_{\bar{S'}}$. Moreover, $u \notin A[v]$ implies that $N^+(z) \cap A[v] = \emptyset$, and thus that $N^+(v) \subseteq S'$, implying that $z \in Z'$.

Let us now consider the reverse implication, and let $S'$ be a solution of $G'$ of size at most $k$ with set of sinks $Z'$ in $G'_{\bar{S'}}$ and prove that $S'$ is a solution of $G$. Let us start with $u \in V(G_{\bar{S'}}) \setminus A[v]$. As $S'$ is a solution of $G'$ and $u \in V(G'_{\bar{S'}})$, there is $uz'$-path $P'$ in $G'_{\bar{S'}}$ where $z' \in Z'$, and this path still exists in $G_{\bar{S'}}$.
As $N^+(z') \cap A[v] = \emptyset$, $z'$ is still a sink in $G_{\bar{S'}}$ and we are done.
Consider now a vertex $u \in V(G_{\bar{S'}}) \cap A[v]$. As $S' \cap A[v] = \emptyset$, there is $uv$-path $P$ in $G_{\bar{S'}}$. If $N^+(v) \subseteq S'$ then $v$ is a sink in $G_{\bar{S'}}$ and we are done. Otherwise, let $w \in N^+(v) \setminus S'$. As $v$ is a SCC of size $1$, $N^+(v) \cap A[v] = \emptyset$, implying that $w \in  V(G_{\bar{S'}}) \setminus A[v]$, and thus according to the previous case $w$ has a path to a sink in $G_{\bar{S'}}$.
\end{proof}

\noindent
The previous reduction rule removes in particular sources and sinks, as
they are SCC's of size one.

\begin{reducao}\label{redSCCG}
Let $U_i$ be a strongly connected component of $G$ with strictly more than $k$
out-neighbors in $G[V\setminus V(U_i)]$. Then we can safely remove
$A[U_i]$.
\end{reducao}

\begin{proof}
Let $G'$ be the graph obtained by removing $A[U_i]$.
Let us first show that $(G,k)$ is a {\em yes}-instance implies that $(G',k)$ is also a {\em yes}-instance.
Let $S$ be a solution of $G$ of size at most $k$ and $Z$ the set of sinks in $G_{\bar{S}}$.
Let $S' = S \setminus A[U_i]$, and $Z'$ the set of sinks in $G'_{\bar{S'}}$. Using the same argument (replacing $A[v]$ by $A[U_i]$) as in the first part of proof of Reduction~\ref{sccsize1}, we get that every $u \in V(G'_{\bar{S'}})$ has a path ot $Z'$ in $G'_{\bar{S'}}$. 

Let us now consider the reverse implication, and let $S'$ be a solution of $G'$ of size at most $k$ with set of sinks $Z'$ in $G'_{\bar{S'}}$ and prove that $S'$ is a solution of $G$. Let us start with $u \in V(G_{\bar{S'}}) \setminus A[v]$. As $S'$ is a solution of $G'$ there is $uz'$-path $P'$ in $G'_{\bar{S'}}$ where $z' \in Z'$, and this path still exists in $G_{\bar{S'}}$. As $N^+(z') \cap A[U_i] = \emptyset$, $z'$ is still a sink in $G_{\bar{S'}}$ and we are done.
Consider now a vertex $u \in V(G_{\bar{S'}}) \cap A[U_i]$. As $S' \cap A[U_i] = \emptyset$, there is $uU_i$-path $P$ in $G_{\bar{S'}}$. 
As $U_i$ has strictly more than $k$ out-neighbors in $G[V\setminus V(U_i)]$, 
there is arc from $U_i$ to $w\in V(G_{\bar{S'}})$ and thus according to the previous case $w$ has a path to a sink in $G_{\bar{S'}}$.
\end{proof}


\section{W[1]-hardness and directed width measures}~\label{sec:digmeasures}
{\sc $k$-KFVD} was shown to be W[1]-hard using a reduction from {\sc
  $k$-Multicolored Independent Set
  ($k$-MIS)}~\cite{cocoon2018}. However, the gadget used in this
reduction to encode each color class has a longest directed path of unbounded
length. First, we remark that it is possible to modify the reduction in
order to prove that {\sc $k$-KFVD} is W[1]-hard even if the input
graph $G$ has longest path length and K-width bounded by constants.

\begin{theorem}\label{reduction}
There is a polynomial-time reduction, preserving the size of the
parameter, from {\sc $k$-MIS} to {\sc $k$-KFVD} such that the
resulting graph has longest directed path of length at most $5$ and
K-width equal to $2$.
\end{theorem}

\begin{proof}
Let $(G',k)$ be an instance of {\sc Multicolored Independent Set}, and let $V^1, V^2, \dots, V^{k}$ be the color classes of $G'$. We construct an instance $(G, k)$ of {\sc Knot-Free Vertex Deletion} with bounded longest path length and K-width as follows.


\begin{enumerate}
\item for each $v_i\in V(G')$, create a directed cycle of size two with the vertices $w_i$ and $z_i$ in $G$;
\item for a color class $V^j$ in $G'$, create one vertex $u_j$;
\item for each vertex $z_i$ in $G$ corresponding to a vertex $v_i$ of
  the color class $V^j$ in $G'$, create an arc from $z_i$ to $u_j$ and
  from $u_j$ to $z_i$.
\item for each vertex $w_i$ in $G$ corresponding to a vertex $v_i$ of
  the color class $V^j$ in $G'$, create an arc from $u_j$ to $w_i$
\item for each edge $e_p=(v_i,v_l)$ in $G'$ create a set $X_p$ with
  two artificial vertices $x^i_p$ and $x^l_p$ and the arcs
  $x^i_px^l_p$ and $x^l_px^i_p$;
\item for each artificial vertex $x^i_p$, create an edge from $x^i_p$
  towards $z_i$ in $G$.
\end{enumerate}

Finally, set $Y_j=\{w_i,z_i\ :\ v_i\in V^j\}\cup \{u_j\}$, $Y_j$ is the set of vertices of $G$ corresponding to the vertices of $G'$ in the same color class $V^j$.
Notice that, the longest path of $G$ has at most $5$ vertices, and for any pair $s,t$ in $V(G)$ there are at most $2$ distinct directed $st$-paths in $G$. 

Now, suppose that now $S'$ is a $k$-independent set with exactly one vertex
of each set $V^j$ of $G'$. By construction, $G$ has $k$ knots which
are $G[Y_1],\dots ,G[Y_k]$. Thus, at least $k$ vertex removals are
necessary to make $G$ free of knots. We set $S=\{z_i~|~v_i \in S'\}$
and show that $G[V \setminus S]$ is knot-free.  For $j=1,\dots ,k$ the
vertex $w_j$ is a sink in $G\setminus S$, and every vertex of $Y_j\setminus S$
still reaches $w_j$. Now, as $S'$ is a
$k$-independent set of $G'$ each set $X_p$ in $G$ is adjacent to at
least one vertex that is not in $S$. Hence, each $X_p$ will still have at
least one arc pointing outside $X_p$, i.e., no new knots are created, and
$G\setminus S$ is knot-free.

Conversely, suppose that $G$ has a set of vertices $S$ of size $k$
such that $G[V\setminus S]$ is knot-free. In particular $S$ has to
contain exactly one vertex of each of the knot $Y_j$, for $j=1,\dots
,k$. Since at least one sink has to be created in order to untie the
knot $Y_j$, and since the only vertices of $Y_j$ with only one
out-neighbor are the $w$'s ones, $S$ has to contain a vertex $z_i$
of each set $Y_1,\dots ,Y_k$. Moreover by deleting one vertex $z_i$ in
a knot $Y_j$, the vertex $w_j$ is turned into sink and every other
vertex of the same knot still has a path to $w_j$. Since $G[V\setminus S]$ is
knot-free, no new knots are created by the deletion of $S$; thus,
every SCC $X_p$ will still have at least one arc pointing outside it. So, we
set $S'=\{v_i~|~z_i \in S\}$. Since each SCC $X_p$ corresponds to an
edge of $G'$, and at least one vertex of each edge of $G'$ is not in
$S'$, the set $S'$ contains no pair of adjacent vertices. Moreover,
$S'$ is composed by one vertex of each knot, which corresponds to a color
of $G'$. Therefore, $S'$ is a multicolored independent set of $G'$.
\end{proof}

Since {\sc $k$-MIS} is W[1]-hard, the following holds.

\begin{corollary}\label{kfw1hard}
{\sc $k$-KFVD} is W[1]-hard even if the input graph has longest
directed path of length at most $5$ and K-width equal to $2$.
\end{corollary}

After the introduction of the notion of directed treewidth
(dtw)~\cite{jhonson2001}, a large number of width measures in digraphs
were developed, such as: cycle rank~\cite{gruber2012} (cr); directed
pathwidth~\cite{barat2006} (dpw); zig-zag number~\cite{Oliveira2013}
(zn); Tree-Zig-Zag number~\cite{Oliveira2016} (Tzn); Kelly-width~\cite{hunter2008} (Kelw);
DAG-width~\cite{berwanger2012} (dagw); D-width~\cite{safari2005} (Dw);
weak separator number~\cite{Oliveira2016} (s);
entanglement~\cite{berwanger2005} (ent); DAG-depth~\cite{ganian2014} (ddp). However, if a graph problem is hard when both the longest
directed path length and the K-width are bounded, then it is hard for
all these measures (see~Figure~\ref{widthhieranch}).

\vspace{.5cm}
\begin{figure}[ht]
    \centering
    \def\svgwidth{1.25\linewidth}
    	  \resizebox{0.85\textwidth}{!}{\input{widthhier.tex}}
  \caption{A hierarchy of digraph width measure parameters. $\alpha \rightarrow \beta$ indicates that $\alpha(G) \leq f(\beta(G))$ for any digraph $G$ and some function $f$. More details about the relationships between these parameters can be found in the references corresponding to each arrow.}
  \label{widthhieranch}
\end{figure}

Therefore, from the reduction presented in Theorem~\ref{reduction} we
can observe that {\sc KFVD} is para-NP-hard with respect to all these
width measures, and {\sc $k$-KFVD} is W[1]-hard even on inputs where
such width measures are bounded.
Thus, it seems to be extremely hard to identify nice width parameters
for which {\sc KFVD} can be solved in FPT-time or even in XP-time.
Fortunately, there remain some parameters for which, at
least, XP-time solvability is achieved. One of them is the
  {\em directed feedback vertex set number} ($dfv$). This invariant is
  an upper bound on the size of a minimum knot-free vertex deletion
  set, so XP-time algorithms are trivial. This parameter is discussed
  in more detail in Section~\ref{sec:feedback}.

Another interesting width parameter for directed graphs $G$ that is
not bounded by a function of the K-width and the length of a longest
directed path is the clique-width of $G$. Courcelle et
al.~\cite{courcelle2000linear} showed that every graph problem
definable in $\textsc{LinEMSOL}$ can be solved in time $f(w)\times
n^{O(1)}$ on graphs with clique-width at most~$w$, when a
$w$-expression is given as input. Using a result of Oum~\cite{Oum08},
the same follows even if no $w$-expression is given.

\begin{proposition}[see \cite{courcelle2012graph}]\label{Xpath}
There is a monadic second-order formula expressing the following
property of vertices $x,y$ and of a set of vertices $X$ of a directed
graph $G$: 
$$\mbox{``$x,y \in X$ and there is a directed path from $x$ to $y$ in the subgraph induced by $X$.''}$$
\end{proposition}

From Proposition~\ref{Xpath} one can show that {\sc KFVD} is
LinEMSOL-definable. Thus Theorem~\ref{the:cwd} holds.

\begin{theorem}\label{the:cwd}
\textsc{KFVD} is {FPT} when parameterized by clique-width of the underlying undirected graph $G$.
\end{theorem}

\begin{proof}
From Proposition~\ref{Xpath}, we can construct (using shortcuts) a formula~$\psi(G,S)$ such that~``$S$ is knot-free vertex deletion set of $G$'' $\Leftrightarrow \psi(G,S)$, as follows:	
	\begin{equation*}\label{MSOL1}
		\begin{split}
			\exists~ Z\subset V~[ \\
			[~\forall ~v\in Z (~\forall ~w\in V (arc(v,w) \implies w\in S)] ~\wedge~~~\\
			[~\forall ~u\in \{V\setminus S\} (~\exists ~s\in Z (~\mbox{there is a directed $\{V\setminus S\}$-path from u to s}~)~]~~~ \\
			]
		\end{split}
	\end{equation*}

Since $\psi(G,Z)$ is an $\textsc{MSOL}_1$-formula, the problem of finding $\min(Z): \psi(G,Z)$ is definable in $\textsc{LinEMSOL}$. Thus we can find $\min(Z)$ satisfying $\psi(G,Z)$ in time $f(cw)\times n^{O(1)}$.
\end{proof}

The fixed-parameter tractability for clique-width parameterization
implies fixed-parameter tractability of {\sc KFVD} for many other
popular parameters. For example, it is well-known that the
clique-width of a directed graph $G$ is at most $2^{2tw(G)+2}+1$,
where $tw(G)$ is the treewidth of the underlying undirected graph
(see~\cite[Proposition~2.114]{courcelle2012graph}).  However, although
Theorem~\ref{the:cwd} implies the FPT-membership of the problem
parameterized by the treewidth of the underlying undirected graph, the
dependence on $tw(G)$ provided by the model checking framework is
huge. So, it is still a pertinent question whether such a parameterized problem admits a 
single exponential algorithm, which is discussed in
Section~\ref{sec:tw}.



\section{The treewidth of the underlying undirected graph as parameter}\label{sec:tw}

Given a tree decomposition $\mathcal{T}$, we denote
by $t$ one node of~$\mathcal{T}$ and by $X_t$ the vertices contained
in the \emph{bag} of $t$. We assume, without loss of generality, that $\mathcal{T}$ is a
\emph{nice} tree decomposition (see~\cite{cygan2015parameterized}),
that is, we assume that there is a special root node~$r$ such that
$X_t = \emptyset$ and all edges of the tree are directed towards $r$
and each node~$t$ has one of the following four types: {\em Leaf},
{\em Introduce vertex}, {\em Forget vertex}, and {\em Join}.

Based on the following results we can assume that we are given a nice tree decomposition of $G$.

\begin{theorem}~\cite{bodlaender2016c}
There exists an algorithm that, given an $n$-vertex graph $G$ and an
integer $k$, runs in time $2^{O(k)}\times n$ and either outputs that
the treewidth of $G$ is larger than $k$, or constructs a tree
decomposition of $G$ of width at most $5k + 4$.
\end{theorem}

\begin{lemma}~\cite{cygan2015parameterized}
Given a tree decomposition $(T, \{X_t\}_{t\in V(T)})$ of $G$ of width
at most $k$, one can in time $O(k^2\cdot \max(|V(T)|,|V(G)|))$ compute
a nice tree decomposition of $G$ of width at most $k$ that has at most
$O(k|V(G)|)$ nodes.
\end{lemma}

Now we are ready to use a nice tree decomposition in order to obtain
an FPT-time algorithm with single exponential dependency on $tw(G)$
and linear with respect to $n$.

\begin{theorem}
{\sc Knot-Free Vertex Deletion} can be solved in $2^{O(tw)} \times n$
time, but assuming ETH there is no $2^{o(tw)}n^{O(1)}$ time algorithm
for {\sc KFVD}, where $tw$ is the treewith of the underlying
undirected graph of the input $G$.
\end{theorem}

\begin{proof}

Let $\mathcal{T} = (T, \{X_t\}_{t \in V(T)})$ be a nice tree
decomposition of the input digraph $G$, with width equal to $tw$.
First, we consider the following additional notation and definitions:
$t$ is the index of a bag of $T$;
$G_t$ is the graph induced by all vertices $v \in X_{t'}$ such that either $t'=t$ or $X_{t'}$ is a descendant of $X_t$ in $T$;
Given a knot-free vertex deletion set $S$, for any bag $X_t$ there is a partition of $X_t$ into $S_t, Z_t, F_t, B_t$ where
\begin{itemize} 
\item $S_t$ (removed) is the set of vertices of $X_t$ that are going to be removed ($S_t=S\cap X_t$);
\item $Z_t$ (sinks) is the set of vertices of $X_t$ that are going to be turned into sinks after the removal of $S$;
\item $F_t$ (free/released) is the set of vertices of $X_t$ that, after the removal of $S$, are going to reach a sink that belongs to $V(G_t)$;
\item $B_t$ (blocked) is the set of vertices of $X_t$ that, after the removal of $S$, are going to reach \emph{no} sink that belongs to $V(G_t)$;
\end{itemize}
Let $Y\subseteq X_t$. We denote by $A_t(Y)$ the set of vertices in $F_t$ that reach some vertex of $Y$ in the graph induced by $V(G_t)\setminus S_t$.

The recurrence relation of our dynamic programming has the signature $C[t, S_t, Z_t, F_t, B_t]$, representing the minimum number of vertices in $G_t$ that must be removed in order to produce a graph such that for every remaining vertex $v$ either $v$ reaches a vertex in $B_t$ (meaning that it may still be released in the future) or $v$ reaches a vertex that became a sink (possibly the vertex itself), where every vertex in $S_t$ is removed, every vertex in $Z_t$ becomes a sink, every vertex in $F_t$ will have a path to a sink in $G_t$, and $S_t, Z_t, F_t, B_t$ form a partition of $X_t$.
Notice that the generated table has size $4^{tw}\times tw\times n$, and when $t=r$, $X_t = \emptyset$ and therefore $C[r, \emptyset, \emptyset, \emptyset, \emptyset ]$ contains the size of a minimum knot-free vertex deletion set of $G_r=G$.

The recurrence relation for each type of node is described as follows.

First, notice that if $v\in Z_t$ and there is an out-neighbor $w$ of
$v$ that is not in $S_t$, there is an inconsistency, i.e. $w$ must be
deleted (must belong to $S_t$). In addition, if $v\in B_t$ but has an
out-neighbor in $Z_t\cup F_t$, there is another inconsistency ($v$ is
not blocked), and if $v\in F_t$ but the removal of $S_t\cup B_t$ turns
$v$ into an isolated vertex, $v$ is not released, and it must belong
to $B_t$. For the inconsistent cases, $C[t, S_t,Z_t,F_t,B_t] =
+\infty$. Such cases can be recognized and treated by simple
preprocessing in linear time on the size of the table. Therefore,
we consider next only consistent cases.

\noindent {\bf Leaf Node:} If $X_t$ is a leaf node then $X_t=\emptyset$. Therefore $$C[t, \emptyset,\emptyset,\emptyset,\emptyset] = 0.$$

\noindent {\bf Insertion Node:} Let $X_t$ be a node of $T$ with a child $X_{t'}$ such that $X_t= X_{t'} \cup  \{v\}$ for some $v \notin X_{t'}$. We have the following:

$$C[t, S_t, Z_t, F_t, B_t] =
\begin{cases}
   1) ~case ~v \in S_t:\\
    \mbox{--} ~~~C[t', S_t\setminus \{v\}, Z_t, F_t, B_t] +1, &\\
   2) ~case ~v \in Z_t:\\   
     \mbox{--} ~~~ \min_{A'\subseteq A_t(v)}\{C[t', S_t,Z_t\setminus \{v\}, F_t \setminus A', B_t\cup A']\}, & \\
   3) ~case ~v \in F_t :\\    
    \mbox{--} ~~~ \min_{A'\subseteq A_t(v)}\{C[t, S_t, Z_t, F_t\setminus \{A'\cup \{v\}\}, B_t\cup A']\}, & \\   
   4) ~case ~v \in B_t:\\
     \mbox{--} ~~~C[t', S_t, Z_t, F_t, B_t\setminus \{v\}] &\\
   \end{cases}.$$
   
\noindent Recall that $A_t(v)$ is the set of vertices in $F_t$ that reach $v$ in the graph induced by $V(G_t)\setminus S_t$, i.e., the set of vertices that can be released by $v$ if it was blocked in $G_{t'}$. Also note that, for simplicity, we consider only consistent cases, thus in case 2 it holds that $N^+(v)\cap X_t\subseteq S_t$, in case 3 it holds that $N^+(v)\cap (Z_t\cup F_t)\neq \emptyset$, and in case 4 it holds that $N^+(v)\cap \{Z_t\cup F_t\}=\emptyset$.

\smallskip

\noindent {\bf Forget Node:} Let $X_t$ be a forget node with a child
$X_{t'}$ such that $X_t = X_{t'} \setminus \{v\}$, for some $v \in
X_{t'}$. The forget node selects the best scenario considering all the 
possibilities for the forgotten vertex, discarding cases that lead to
non-feasible solutions. In this problem, unfeasible cases are identified
when the forgotten vertex $v$ of $X_{t'}$ was blocked and reached no
other node in $B_t$. Hence:

\begin{itemize}
\item If $N^+(v)\cap B_{t'} \neq \emptyset$ then
$$C[t, S_t, Z_t, F_t, B_t] = \min
\begin{cases}
   C[t', S_t \cup \{v\}, Z_t, F_t, B_t],\\
    C[t', S_t, Z_t \cup \{v\}, F_t, B_t],\\
   C[t', S_t, Z_t, F_t \cup \{v\}, B_t],\\   
    C[t', S_t, Z_t, F_t, B_t \cup \{v\}]\\
   \end{cases}.$$
\item If $N^+(v)\cap B_{t'} = \emptyset$ then
$$C[t, S_t, Z_t, F_t, B_t] = \min
\begin{cases}
   C[t', S_t \cup \{v\}, Z_t, F_t, B_t],\\
   C[t', S_t, Z_t \cup \{v\}, F_t, B_t],\\
   C[t', S_t, Z_t, F_t \cup \{v\}, B_t],\\   
\end{cases}.$$
\end{itemize}

\noindent {\bf Join Node:} Let $X_t$ be a join node with children
$X_{t_1}$ and $X_{t_2}$, such that $X_t = X_{t_1} = X_{t_2}$. For any
optimal knot-free vertex deletion set $S$ of $G$ it holds that $V(G_t)\cap S
= \{V(G_{t_1})\cap S\} \cup \{V(G_{t_2})\cap S\}$. Clearly, if
$S_t\subseteq S$ then we can assume that $S_t=S_{t_1}=S_{t_2}$. In
addition, $Z_t=Z_{t_1}=Z_{t_2}$ otherwise we will have an
inconsistency. Also note that a vertex is released in $G_t$ if it
reaches a vertex (possibly the vertex itself) that is released either
in $G_{t_1}$ or $G_{t_2}$. Thus:
$$C[t, S_t, Z_t, F_t, B_t] = \min_{\forall F', F''}
\{ C[t_1, S_t, Z_t, F', B']+C[t_2, S_t, Z_t, F'', B'']\}- |S_t|,\textrm{~where~} A_t(F'\cup F'') = F_t.$$

Note that $A_t(F'\cup F'')$ is the set of vertices that either are
released in $G_{t_i}$ ($i\in\{1,2\}$) or can be released in $G_t$ by
vertices of $F'\cup F''$, even if they are blocked in both $G_{t_1}$ and
$G_{t_2}$; this can occur, for example, if a blocked vertex $v$
reaches another blocked node $w$ in $G_{t_1}$, and in $G_{t_2}$ vertex $w$ is
released.

Now, in order to run the algorithm, one can visit the bags of
$\mathcal{T}$ in a bottom-up fashion, performing the queries described
for each type of node.  Since the reachability between the vertices of
a bag can be stored in a bottom-up manner on $\mathcal{T}$, one can
fill each entry of the table in $2^{O(tw)}$ time, and as the table has
size $2^{O(tw)}\times n$, the dynamic programming can be performed in
time $2^{O(tw)}\times n$.

Regarding correctness, let $S^*$ be a minimum knot-free vertex deletion set of a digraph $G$ with a tree decomposition $\mathcal{T}$. Let $S_t^*,Z_t^*,F_t^*,B_t^*$ be a partition of the vertices of $X_t$ into removed, sinks, released and blocked, with respect to $G_t$ after the removal of $S^*$. Note that $S_t^* =X_t\cap S^*$. 

\medskip

\noindent {\bf Fact 1.} {\em There is no vertex $w\in V(G_t)\setminus
  X_t$ such that $w$ reaches a vertex $v\in B_t^{*}$ in
  $G[V(G_t)\setminus S_t]$ and $w\in S^*$. Otherwise, since every
  vertex in $B_t^{*}$ will reach a sink that is not in $G_t$, by Observation~\ref{lem:re-add} one can
  remove from $S^*$ every vertex that reaches $B_t^{*}$ in
  $G[V(G_t)\setminus S_t]$, obtaining a subset of $S^*$ which is also
  a knot-free vertex deletion set, contradicting the fact that S
  is minimum.}
  
\medskip

This fact implies that the paths considered to compute
$A_t(v)/A_t(F'\cup F')$' can in fact be used to release blocked
vertices. Similarly, Fact 2 also holds.

\medskip

\noindent {\bf Fact 2.} {\em Let $\widehat{S}$ be a set for which the
  minimum is attained in the definition of
  $C[t,S_t^*,Z_t^*,F_t^*,B_t^*]$. Then $\widehat{S}\cup (S^*\setminus
  V(G_t))$ is also a solution (which is minimum) for {\sc
    KFVD}. Otherwise, from $\widehat{S}\cup (S^*\setminus V(G_t))$ we
  can also obtain a knot-free vertex deletion set smaller than $S^*$,
  which is a contradiction.}
  
\medskip

Fact 2 implies that we have stored enough information. At this point,
the correctness of the recursive formulas is straightforward.

Finally, to show a lower bound based on ETH, we can transform an
instance $F$ of {\sc 3-SAT} into an instance $G_F$ of {\sc KFVD} using
the polynomial reduction presented in~\cite[Theorem~4]{cocoon2018},
obtaining in polynomial time a graph with $|V| = 2n+2m$, and so
$tw=O(n+m)$. Therefore, if {\sc KFVD} can be solved in
$2^{o(tw)}|V|^{O(1)}$ time, then we can solve 3-SAT in
$2^{o(n+m)}(n+m)^{O(1)}$ time, i.e., ETH fails.
\end{proof}

\section{The size of a minimum directed feedback vertex set as parameter}\label{sec:feedback}

Recall
that {\sc $k$-KFVD} is $W[1]$-hard (for fixed K-width and longest
directed path) and that, as noticed in Observation~\ref{viability}, we
can assume $k<dfv(G)$. This motivates us to determining the status
of {\sc $dfv$-KFVD}. In this section, we present two FPT-algorithms. Both with the size
of a minimum directed feedback vertex set as parameter but with an aggregate
parameter, the K-width, $\kappa (G)$, for the first one and the length
of a longest directed path, $p(G)$, for the second one. Since finding a minimum
directed feedback vertex set $F$ in $G$ can be solved in
FPT-time (with respect to $dfv$)~\cite{chen2008fixed}, we consider that $F$, a minimum DFVS, is given. Namely, we show that both {\sc $(dfv,\kappa)$-KFVD} and {\sc $(dfv,p)$-KFVD} are FPT.

At this point, we need to define the following variant of {\sc KFVD}.
\medskip

\noindent    \fbox{
        \parbox{.965\textwidth}{
\noindent
{\sc \textsc{Disjoint Knot-Free Vertex Deletion (Disjoint-KFVD)}}

\noindent
\textbf{Instance}: A directed graph $G=(V,E)$; a subset $X\subseteq V$; and a positive integer $k$.


\noindent
\textbf{Question}: Determine if $G$ has a set $S \subset V(G)$ such that $|S| \leq k$, $S\cap X=\emptyset$ and $G[V\setminus S]$ is knot-free.
}
}
\medskip

\noindent
We call {\it forbidden vertices} the vertices of the
set $X$. It is clear that {\sc Disjoint-KFVD}
generalizes {\sc KFVD} by taking $X=\emptyset$.  

Let us now define two more steps that are FPT parameterized by $dfv$ and that will be used for both $(dfv,\kappa)$-KFVD and $(dfv,p)$-KFVD. The next step will allow us to consider that the vertices of $F$ are forbidden.
We need the following straightforward observation.
\begin{obs}\label{redGuessed}
Let $(G,k)$ be an instance of {\sc KFVD} and $v \in V(G)$. 
\begin{itemize}
    \item if $(G,k)$ is a {\em yes}-instance and there exists a solution $S$ with $v \in S$, then $(G \setminus \{v\},k-1)$ is a {\em yes}-instance
    \item if $(G \setminus \{v\},k-1)$ is a {\em yes}-instance then $(G,k)$ is a {\em yes}-instance
\end{itemize}
\end{obs}

\begin{branch}[On the directed feedback vertex set $F$]\label{branch:fvs} Let $(G, F, k)$ be an
  instance of {\sc $dfv$-KFVD}. In time $3^{dfv}\times O(n)$ we can build $3^{dfv}$ instances $(G^i, F_1^i, X^i, k^i)$ of {\sc $dfv$-Disjoint-KFVD} as follows. We consider all
  possible partitions of $F$ into three parts: $F_1$, the set of  vertices of $F$ that will not be removed (i.e., they become forbidden); $F_2$, the set of vertices in $F$ that will be removed; and $F_3$, the set of vertices in $F$ that will be turned into sinks. For each such a partition (indicated by the index $i$), we remove the set $Y^i = F_2^i \cup N^+(F_3^i)$ of vertices and we apply exhaustively Reduction Rules~\ref{sccsize1} and~\ref{redSCCG} (see Section~\ref{sec:redrules}). We denote by $G^i$ the obtained graph, 
  $X^i = F_1^i$,
  and $k^i=k-|Y|$.  
\end{branch}

\noindent
 According to Observation~\ref{redGuessed}, it is clear that $(G,F,k)$ is a {\em yes}-instance of {\sc $dfv$-KFVD} if and only if one of the instances $(G^i,F_1^i,X^i,k^i)$, $1 \leq i \leq 3^{dfv}$, of {\sc $dfv$-Disjoint-KFVD} is a {\em yes} instance.
 Since there are at most $3^{dfv}$ partitions of $F$, the branching reduction can be performed in FPT time. Although at this point $X^i = F_1^i$, in the next steps some vertices of $V(G)\setminus F_1$ may become forbidden and therefore should be added to $X^i$. Also, from this point forward, we assume that we are given an instance $(G,F_1,X,k)$ of {\sc $dfv$-Disjoint-KFVD}. 

Notice that after applying Reduction Rule~\ref{sccsize1} (Section~\ref{sec:redrules}), each strongly
connected component of $G$ is at least of size two. Thus, each of them
must contain at least one cycle; therefore, the number of strongly
connected components of $G$ is bounded by $dfv$. Moreover, for any strongly connected component $U$ of $G$, Reduction Rule~\ref{redSCCG} gives an upper bound for the number of vertices in $N^+(V(U))$ (i.e., vertices that are not in $U$ but it is out-neighbour of some vertex in $U$). 
This implies that $G$ has at most $dfv \times k\le dfv^2$ such vertices between its strongly connected components. This observation leads to a branching rule.


\begin{branch}[On strongly connected components] \label{branch:scc}
 Let $S_H$ be the set of
vertices that are extremities of arcs between the strongly connected
components of $G$. We have $|S_H|\le 2\times dfv\times k\le 2\times dfv^2$ and we can branch in FPT-time trying all possible partitions of $S_H$ into two sets: $S_1$, the set of vertices to be deleted in $G$ such that $|S_1|\leq k$; and $S_2=S_H\setminus S_1$, the set of vertices marked as forbidden, and then added into $X$.
\end{branch}

Notice that this step involves a $2^{|S_H|}$ branching. At this point, we may consider that we have an instance $(G,F,X,k)$ where $F \subseteq X$ and such that for any arc $uv$ between two SCC's $U_i$ and $U_j$, $\{u,v\} \subseteq X$. We call such an instance as a \emph{nice} instance.

\begin{lemma}[After cleaning of Branching~\ref{branch:scc}]
\label{weakToKnot}
If there is an algorithm running in time $g(dfv)\times poly(n)$ for {\sc $dfv$-Disjoint-KFVD} restricted to nice instances that are strongly connected, then there is an FPT algorithm running in time $g(dfv)\times poly(n)\times c.n.log(dfv)$ (where $c$ is a constant) to solve {\sc $dfv$-Disjoint-KFVD} for any nice instance. 
\end{lemma}

\begin{proof}
Let $(G,F,X,k)$ be a nice instance and $S$ be a solution.
Let $\mathcal{U}=\{U_1,\dots,U_s\}$ be the partition of $V(G)$ where each $U_i$ is an SCC, and let $\mathcal{K}=\{U_i :\mbox{ $U_i$ is a knot} \}$. Without loss of generality  we can assume that $\mathcal{K}=\{U_1,\dots ,U_t\}$ for some $t\le s$.
Let $S_i = S \cap U_i$. 
Notice that if $S$ is a solution then for any $i \in [t]$, $S_i$ is a solution of $(G[U_i],F \cap U_i,X \cap U_i, |S_i|)$. Moreover, for any solutions $S'_i$ to $(G[U_i],F \cap U_i,X \cap U_i, |S'_i|)$ where $\sum_{i=1}^t |S'_i| \le k$, $S' = \bigcup_{i=1}^t S'_i$ will be a solution to $(G,F,X,k)$ because vertices of some $U_j \notin \mathcal{K}$ will still have a path to a set $U_i \in \mathcal{K}$ in $G_{\bar{S'}}$ since any arc between two SCC's has forbidden endpoints.
Thus, given a nice instance $(G,F,X,k)$ and an algorithm $A$ for a nice instance restricted to one SCC, for any $U_i \in \mathcal{K}$ we perform a binary search to find the smallest $k_i$ such that $A(G[U_i],F \cap U_i,X \cap U_i, k_i)$ answers {\em yes}, and we answer {\em yes} iff $\sum_{i=1}^t k_i \le k$.
\end{proof}

From Lemma~\ref{weakToKnot}, we may assume that we have an instance $(G,F,X,k)$ such that $F \subseteq X$ and $G$ is strongly connected (there is only one SCC). We call such an instance as a \emph{super nice} instance. 






\subsection{Combining DFVS-number and K-width}
In this section, we prove that $(dfv,\kappa)$-Disjoint-KFVD restricted to super nice instances is FPT.

Let $F = \{v_1,\dots ,v_{dfv}\}$. For every pair of integers
$i,j$ with $1\le i, j\le dfv$ we define $H_{i,j}$ as the
\textit{$(i,j)$-connectivity set}, that is, the set of vertices which
are contained in a directed path from $v_i$ to $v_j$ in the induced
subgraph $G[V\setminus (F\setminus \{v_i,v_j\})]$ (if $i=j$ then
$H_{i,i}$ is the set of vertices contained in a cycle in $G[V\setminus
  (F\setminus \{v_i\})]$). Let us define a set $B$ on which we will
later branch in a way to ensure connectivity among different
connectivity sets. We start with $B = \{\emptyset\}$, and then, for
each possible pair of connectivity sets $H_{i,j}$, $H_{i',j'}$ we increase $B$ as follows:
\smallskip
\begin{enumerate}
\item[i)] add $N^+(H_{i,j} \setminus H_{{i',j'}}) \cap H_{{i',j'}}$ to $B$.

\item[ii)] add $N^+(H_{i,j} \cap H_{{i',j'}}) \cap (H_{{i',j'}}
  \setminus H_{{i,j}})$ to $B$.

\item[iii)] add 
$N^+(H_{i',j'} \setminus H_{i,j}) \cap H_{i,j}$ to $B$.

\item[iv)] add $N^+(H_{{i',j'}} \cap H_{i,j}) \cap (H_{{i,j}}
  \setminus H_{{i',j'}})$ to $B$.
\end{enumerate}
\smallskip
For a given pair of connectivity sets, in each of the items $i)$, $ii)$, $iii)$ and $iv)$ the number of added vertices to $B$
is at most $\kappa$. For instance,let  $y_1,\dots ,y_l$ be the vertices added by item $i)$, where each $y_s \in N^+(H_{i,j} \setminus H_{{i',j'}}) \cap H_{{i',j'}}$. By definition, there exist vertices $x_1,\dots ,x_l$ of $H_{i,j} \setminus H_{{i',j'}}$ such that $x_sy_s$ are arcs of $G$ for $s=1,\dots ,l$. Notice that while the $y_s$'s are distinct, the $x_s$'s are not forced to be so. For any $s \in \{1,\dots ,l\}$, there exists a path $P_s$ in $H_{i',j'}$ from $y_s$ to $v_{j'}$, and such a path does not intersect $H_{i,j} \setminus H_{{i',j'}}$. In the same way, by finding a path $Q_s$ from $v_i$ to $x_s$ for every $s \in \{1,\dots ,l\}$, we form $l$ distinct paths $Q_sP_s$ from $v_i$ to $v_{j'}$, implying 
$l \le \kappa$, the K-width of $G$.
 So, as there are $dfv^2$ different connectivity sets, at
the end of the process $B$ contains at most $\kappa \times~dfv^4$ vertices.  Figure~\ref{figcs} shows
examples of vertices to be added in $B$ regarding the interaction of
two different connectivity sets.

\begin{figure}[!ht]
	\centering
		\includegraphics[width=.9\textwidth]{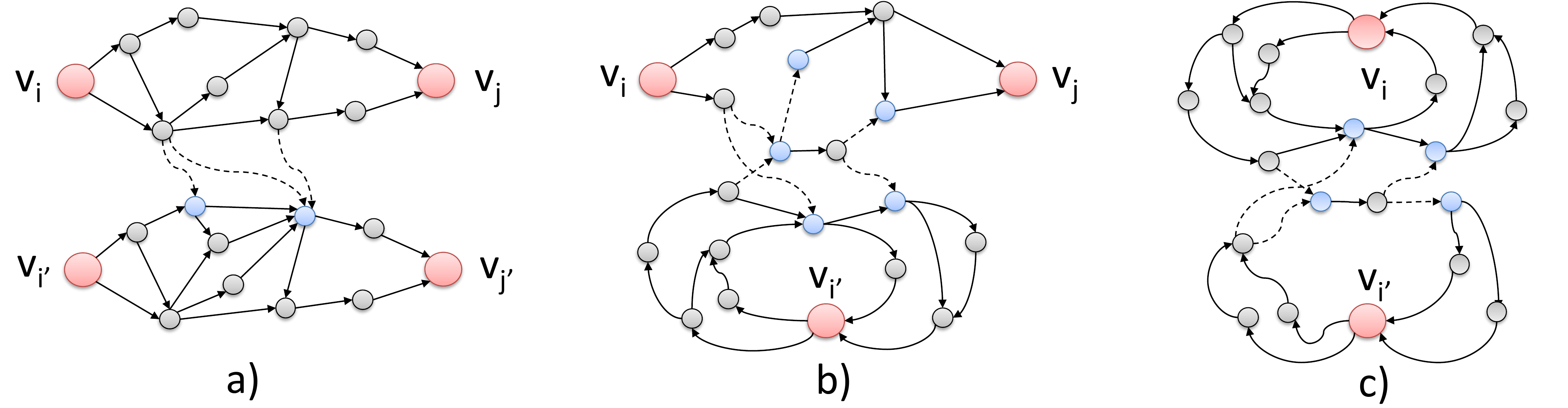}
		\caption{a) two connectivity sets with no intersection. b) an intersection with two vertices belonging to both connectivity sets. c) two connectivity sets $H_{i,j}$ with $i=j$. Vertices to be added in $B$ are marked in blue.}  
  \label{figcs}
\end{figure}

Next we establish our last branching rule.

\begin{branch}[On the connectivity sets]\label{feedSet}
 We branch by partitioning $B$
into three parts: $B_1$, the set of vertices that will not be removed
(ie. they become forbidden); $B_2$, the set of vertices that will be
removed; and $B_3$, the set of vertices that will become sinks.
Since $|B| \le  \kappa\times dfv^4$,  we branch at most $3^{\kappa.dfv^4}$ times.
\end{branch}

At this point, without loss of generality, one can assume that none of the above branching and reductions rules are applicable. Hence, the analysis boils down to
the case where $F \cup B \subseteq X$, meaning that all the vertices of $F \cup B$ are forbidden to be deleted or become sinks, and $G$ is strongly connected.


\begin{obs}[The consequences of Branching~\ref{feedSet}]\label{obsConSet}
Let $G$ be a graph for which no Reduction Rules~\ref{sccsize1} and~\ref{redSCCG} or Branching Rules~\ref{branch:fvs} to~\ref{feedSet} can be applied. Let
$H_{{i,j}}$ and $H_{{i',j'}}$ be two different connectivity arc sets
in $G$. If there is an arc from $H_{{i,j}} \setminus H_{{i',j'}}$ to
$H_{{i',j'}} \setminus H_{{i,j}}$ or $H_{{i,j}} \cap H_{{i',j'}}$ to
$H_{{i',j'}} \setminus H_{{i,j}}$ in $G[H_{{i,j}} \cup H_{{i',j'}}]$,
then the head vertex of such an arc is a forbidden vertex.
\end{obs}

We now aim to show that, for any vertex $v^*$ such that $v^*$ can be
turned into a sink, that is, $N^+(v^*)\cap X=\emptyset$ and $d^+(v^*) \le k$,
the deletion of $N^+(v^*)$ is sufficient for $G$ to become
knot-free.


\begin{lemma}\label{sinkv*}
  Let $(G, F, X, k)$ be an instance of {\sc $(dfv, \kappa)$-Disjoint-KFVD} such
  that $G$ is strongly connected and none of the branching and reduction rules can be applied. If there is a vertex $v^*$ with no forbidden out-neighbors, then $G[V
    \setminus N^+(v^*)]$ is knot-free.
\end{lemma}

\begin{proof}
Let $(G, F, k, X)$ and $v^*$ as stated. Denote by $G'$ the resulting graph, i.e, $G' = G[V \setminus  N^+(v^*)]$. For contradiction, assume that $G'$ contains a knot $K$. 
 As $G$ is strongly connected, $K$ was not a knot in $G$, implying that there exists an arc $xy$ of $G$ such that $x\in V(K)$ and $y\in N^+(v^*)$. Notice that $v^* \notin F$ since vertices from $F$ cannot become sinks and $y \notin X$, since $y$ has to be deleted in order to $v^*$ to become a sink.
 Let us now define a connectivity set containing both $y$ and $v^*$.
 Let $s$ be any source of the DAG $G[V \setminus F]$ such that there is a $sv^*$ path in $G[V \setminus F]$, 
 and let $z$ be any sink of $G[V \setminus F]$ such that there is a $yz$ path in $G[V \setminus F]$.
 As $G$ is strongly connected, there exist arcs $v_is$ and $zv_j$ where $\{v_i,v_j\} \subseteq F$ and we get
 that $\{v^*,y\} \subseteq H_{i,j}$.
 Notice that $i=j$ is possible. Similarly, as $G[V(K)]$ is strongly connected, it contains a cycle $C'$ containing $x$ and thus there exists a connectivity set $H_{k,l}$ containing a path $P$ from $v_k$ to $v_l$ which is a subpath of $G[V(K)]$ containing $x$, and with $\{v_k,v_l\} \subseteq V(K)$. Notice first that $v^*,y\notin F$. In addition, $v^*$ is not a vertex of $H_{k,l}$, otherwise there would exist a path $P'$ from $v_k$ to $v^*$ containing no vertex of $F\setminus \{v_k\}$, which is not possible. Indeed, either $V(P') \cap N^+(v^*) = \emptyset$ and we would get that $K$ is not a knot, or $V(P') \cap N^+(v^*) \neq \emptyset$, implying that there is a cycle with no vertex of $F$.
 Thus, as $y$ was not a forbidden vertex, it means that $y\notin H_{k,l}$ otherwise the arc $v^*y$ would go from $H_{i,j}\setminus H_{k,l}$ to $H_{i,j}\cap H_{k,l}$ and $y$ should be forbidden by Branching~\ref{feedSet} item $i)$.
Then we have $y\in H_{i,j}\setminus H_{k,l}$.
Similarly, we have $x\notin H_{i,j}\cap H_{k,l}$, otherwise by item $ii)$ of 
Branching~\ref{feedSet}, vertex $y$ would be forbidden. Finally $x\in H_{k,l}\setminus H_{i,j}$ and $y \in H_{i,j}\setminus H_{k,l}$, since $(H_{i,j}\setminus H_{k,l}) \subseteq H_{i,j} $, and by item $iii)$ of  Branching~\ref{feedSet}, vertex $y$ would again be a forbidden vertex, a contradiction.
\end{proof}

In conclusion, by Lemma~\ref{sinkv*}, 
we can find in polynomial time the optimum solution for $G$: we choose a vertex $v^*$ with minimum out-degree.


\begin{theorem}
{\sc Knot-free Vertex Deletion} can be solved in $2^{O(\kappa dfv^{5})}\times n^{O(1)}$.
\end{theorem}
\begin{proof}
Let us now compute the running time of the overall algorithm. First notice that applying Branchings 1 and 2 results in $3^{dfv}\times~2^{2dfv^2}$ branches. Branching 3 can be done in time $3^{\kappa.dfv^4}$, but may re-create several SCC's, forcing us to apply again Branching 2 and reduction rules again, but decreasing $k$.
This implies that the total running time is $3^{dfv}\times~(2^{2dfv^2}3^{\kappa.dfv^4})^k \times n^{O(1)}$, thus the result holds.\end{proof}

\subsection{Combining DFVS-number and length of a longest directed path}
In this subsection we investigate the length of a longest path and $dfv(G)$ as aggregate parameters.


\begin{lemma}
$(dfv,p)$-Disjoint-KFVD on super nice instances can be solved in 
$2^{O(dfv^3)}p^{O(dfv)}\times n^{O(1)}$.
\end{lemma}

\begin{proof}
Let $(G,F,X,k)$ be a super nice instance. 
Recall that the directed feedback vertex set $F$ is a set of forbidden vertices ($F \subseteq X$) and $G$ is strongly connected. 
The proof is by induction on $|F|$.
If $|F| = 1$, then, for any vertex $v$ of $V(G) \setminus F$ that can be turned into a sink, $N^+(v)$ will be a solution set for $G$. Therefore, the optimum solution can be found in polynomial time.
Assume now that $|F|\ge 2$ and denote $F$ by $\{v_1,\dots ,v_{dfv}\}$.
As $G$ is strongly connected, there exists a path $P_1$ of length at most $p$ from $v_1$ to $v_2$ and a path $P_2$ of length at most $p$ from $v_2$ to $v_1$. Denote by $C$ the digraph $G[V(P_1)\cup V(P_2)]$; it is strongly connected, contains $v_1$ and $v_2$ and at most $2p$ vertices. Since the number of vertices in $C$ is bounded, we may branch $2p+1$ times by trying to guess a vertex that will be deleted in $C$.
Each time a vertex of $C$ will be guessed as deleted, the parameter $k$ will decrease by one. So, $k$ will decrease in all branches, except in the one where we guess that no vertex is deleted, and then where all the vertices of $C$ are forbidden.
In this case, $C$ is a strongly connected component whose vertices are all forbidden and containing at least two vertices of $F$. So, we contract $C$ to obtain a new instance $G'$.
Formally, we remove $V(C)$ from $G$, add a new vertex $v_C$, and for all vertices 
of $G\setminus C$ having at least one in-neighbor (resp. out-neighbor) in $C$, we add an arc from $v_C$ (resp. to $v_C$) to this vertex. Let $F'$ be the set $\{v_C\}\cup F\setminus V(C)$ and notice that $F'$ is a directed feedback vertex set of $G'$ and that $|F'|<|F|$.
Similarly, let $X'$ be the set $(X \setminus V(C)) \cup \{v_C\}$. We claim that both instances $(G,F,k,X)$ and $(G',F',k, X')$ are equivalent. Indeed, it suffices to notice that as $V(C)$ contains only forbidden vertices in $G$ and that $v_C$ is forbidden in $G'$, then any solution to the KFVD problem for $G$ is a solution of $G'$, and conversely.
Then, we apply Branchings 1 and 2 to obtain a super nice instance equivalent to  $(G',F',k, X')$, and we can apply the induction hypothesis.\\
So at each branching, either the parameter $k$ decreases by at least one or the size of $F$ decreases by at least one. As both values are bounded above by $dfv$, we branch consecutively at most $2dfv$ times. And since Branching rules 1 and 2 create at most 
$3^{dfv}\times 2^{2dfv^2}$ branches, and branching on cycle $C$ creates $2p+1$ branches, the total number of branches is $(3^{dfv}\times~2^{2dfv^2} \times (2p+1))^{2dfv}= 2^{O(dfv^3)}p^{O(dfv)}$, and we get the desired running time.  
%
%
\end{proof}

Given that we can obtain a super nice instance in $2^{O(dfv^3)}\times n^{O(1)}$, it holds that {\sc Knot-free Vertex Deletion} can be solved in time $2^{O(dfv^3)}p^{O(dfv)}\times n^{O(1)}$.


\if 10 

\section{Inaproximability}

\begin{theorem}
3SAT $\propto$ KFVD. Let $F$ be a instance of 3SAT and $G$ a graph obtained by $F$, if $F$ is satisfiable, $opt(G) \leq n$, otherwise  $opt(G) \geq n+p(n)-1$.
\end{theorem}

\begin{proof}
Let $F$ be an instance of 3-SAT with $n$ variables and having at most 3 literals per clause. From $F$ we build a graph $G_F = (V, E)$ which will contain a set $S\subseteq V(G)$ such that $|S|=n$ and $G_F[V\!\setminus\!S]$ is knot-free if and only if $F$ is satisfiable. The construction of $G_F$ is described below:

\begin{enumerate}
 \item For each variable $x_i$ in $F$, create a directed cycle with two vertices (``variable cycle''), $Tx_i$ and $Fx_i$, in $G_F$.

 \item For each clause $C_j$ in $F$ create a directed clique with $p(n) \geq 3$ vertices (``clause cycle''), where  each literal of $C_j$ has a corresponding vertex in the cycle.

\item For each vertex $v$ that corresponds to a literal of a clause $C_j$, create an arc from $v$ to $Tx_i$ if $v$ represents the positive literal $x_i$, and create an arc from $v$ to $Fx_i$ if $v$ represents the negative literal $\bar{x_i}$.	
\end{enumerate}


Suppose that $F$ admits a truth assignment $A$. We can determine a set of vertices $S$ with cardinality $n$ such that $G_F[V\!\setminus\!S]$ is knot-free as follows. For each variable of $F$, select a vertex of $G_F$ according to the assignment $A$ such that the selected vertex represents the opposite value in $A$, i.e., if the variable $x_i$ is true in $A$, $Fx_i$ is included in $S$, otherwise $Tx_i$ is included in $S$. Since each knot corresponds to a variable cycle, it is easy to see that $G_F[V\!\setminus\!S]$ has exactly $n$ sinks. Therefore, since $A$ satisfies $F$, at least one vertex corresponding to a literal in each clause cycle will have an arc towards a sink (vertex that matches the assignment). Thus $G_F[V\!\setminus\!S]$ will be knot-free.

Conversely, suppose that $G_F$ contains a set $S$ with cardinality $n$ such that $G_F[V\!\setminus\!S]$ is knot-free. By construction $G_F$ contains $n$ knots, each one associated with a variable of $F$. Hence, $S$ has exactly one vertex per knot (one of $Tx_i$, $Fx_i$). As each cycle of $G_F[V\!\setminus\!S]$ corresponds to a clause of $F$, and $G_F[V\!\setminus\!S]$ is knot-free, each cycle of $G_F[V\!\setminus\!S]$ has at least one out-arc pointing to a sink. Thus, we can define a truth assignment $A$ for $F$ by setting $x_i = true$ if and only if $Tx_i \in V\setminus S$. Since at least one vertex corresponding to a literal in each clause cycle will have an arc towards a sink, we conclude that $F$ is satisfiable.

Finally, if $F$ is not satisfiable, by the arguments above any deletions in the variable gadget will at the and create at least one new knot. The new knot is a directed clique of size $p(n)$, therefore, the optimum way to untie the new knot is by deleting $p(n)-1$ vertices.
\end{proof}

\fi



\if 10 

\section{No kernel / distance to CC of size 2}
We say that a digraph $G=(V,E)$ has connected components of size at most $2$ (and we denote by $cc \le 2$) iff the underlying undirected graph of $G$ (that have an edge $\{u,v\}$ iff $uv$ or $vu$ is an arc of $E$) has
connected components of size at most $2$.
We say that a set $X \subseteq V$ is a modulator to connected components of size at most $2$ iff $G[V \setminus X]$ has $cc \le 2$.
We denote by $f_{\le 2}(G)$ the minimum size of such a modulator.
\fixme{
\begin{itemize}
\item explain here or in the intro success of modulators in kernelization complexity for VC
\item voir si param c'est min size ou c'est juste size d'un modulator
\item check if $G[..]$ and [t] is defined, and if command for NP
\item add material for cross composition
\item uniformize with other reduction from sat : no need to create a triangle for clauses, and name of variable vertices
\item uniformize with notations : I used D and S like I like
\end{itemize}
}

\begin{theorem}
KFVD parameterized bye the minimum size of a modulator to $cc \le 2$ does not admit a polynomial kernel unless $NP \subseteq coNP / poly$
\end{theorem}
\begin{proof}
  Let us define a cross composition from $3-SAT$.
  We consider $t$ instances of  $3-SAT$, and using an appropriate equivalence relation we can assume that
  all instances have the same number of variables, denoted by $n$, and the same number of clauses, denoted by $m$.
  Without loss of generality we can also assume that $t$ is a power of two by adding dummy no instances.
  For any $\ell \in [t]$ and $j \in [m]$, we denote by $\tilde{C}^\ell_j$ the $j^{th}$ clause of instance $I^\ell$.
  Our objective is to define a digraph $G=(V,E)$ and an integer $k$ such that $(G,k)$ is a {\em yes}-instance iff there exists $\ell \in [t]$ such that $I^\ell$
  is a {\em yes}-instance, and such that $f_{\le 2}(G) \le poly(n,m,log(t))$.

  We partition $V$ into $T \cup C \cup V$.
  Set $V$ has vertices $\{v_{(i,T)},v_{(i,F)}\}$ and arcs $v_{(i,T)}v_{(i,F)}$, $v_{(i,F)}v_{(i,T)}$, for $i \in [n]$.
  Observe that $G[V]$ is the disjoint union of $n$ digons.
  Set $C$ is itself partitioned into $\bigcup_{\ell \in [t]}C^\ell$.
  Now, for any $\ell \in [t]$, we will define $G$ such that $G[V \cup C^\ell]$ is the graph obtained by applying our previous reduction~\cite{} from $3-SAT$
  to instance $I^{\ell}$. More precisely, each $C^\ell$ is partitioned into $\bigcup_{j \in [m]} C^\ell_j$.
  Set $C^\ell_j$ has vertices $\{c^\ell_{(j,1)},c^\ell_{(j,2)}\}$ and arcs $c^\ell_{(j,1)}c^\ell_{(j,2)}$, $c^\ell_{(j,1)}c^\ell_{(j,2)}$, for $j \in [m]$.
  Observe that $G[C^\ell]$ is the disjoint union of $m$ digons.
  Now we add arcs from $C$ to $V$ as follows. For any $\ell \in [t]$, for each literal of clause $\tilde{C}^\ell_j$, create an arc from $c^\ell_{(j,1)}$ to $v_{(i,T)}$ if the literal is $x_i$,
  of from $c^\ell_{(j,1)}$ to $v_{(i,F)}$ if the literal is $\bar{x_i}$.
  Finally, $T$ has vertices $\{t_{(u,0)},t_{(u,1)}\}$ and arcs $t_{(u,0)}t_{(u,1)}$, $t_{(u,1)}t_{(u,0)}$, for $u \in [log(t)]$.
  Again, $G[T]$ is the disjoint union of $log(t)$ digons.
  For any $\ell \in [t]$, we denote by $(b^\ell_1,\dots,b^\ell_{log(t)})$ the binary representation of $\ell$.
  Let $\ell \in [t]$. For any $j \in [m]$ we add arcs from $c^\ell_{(j,1)}$ to $\{t_{(u,b^\ell_u)}, u \in [log(t)]\}$. Informally, each $c^\ell_{(j,1)}$ have $log(t)$ arcs pointing to the binary representation of $\ell$.   This concludes the description of $G$. Notice that $X = T \cup V$ is a modulator $cc \le 2$ as $G[V \setminus X]=G[C]$ is a collection of $mt$ digons, implying that $f_{\le 2}(G) \le |X| = 2n+2log(t)$. We define $k=log(t)+n$, and it only remains to prove that $(G,k)$ is a {\em yes}-instance iff there exists $\ell \in [t]$ such that $I^\ell$ is a {\em yes}-instance.

  \textbf{$\Leftarrow$} Suppose $I^{l_0}$ is satisfied by an assignment $A$ and let us define a set $D$ of vertices to delete.
  We start by adding vertices $\{b^{\ell_0}_1,\dots,b^{\ell_0}_{log(t)}\}$ to $D$.
  Now, for each variable, if the variable $x_i$ is true in $A$ we add $v_{(i,F)}$ in $D$, and otherwise we add $v_{(i,T)}$ in $D$. Observe that $|D|=k$.
  Let us now verify that any vertices of $V \setminus D$ is connected to a sink in $G[V \setminus D]$.
  Notice first that the $log(t)$ vertices of $T\setminus D$ become sinks,
  and that all vertices of $C^{\ell}$ for $\ell \neq \ell_0$ are connected to at least one of these sinks in $G[V \setminus D]$.
  Observe also that the $n$ vertices of $V \setminus D$ become sinks, and that for any $j \in [m]$, at least one of three out arcs from $c^\ell_{(j,1)}$ to $V$ still exists
  in $G[V \setminus D]$ (as $A$ satisfies each clause), implying that vertices of $C^\ell_j$ have a path to a sink in $G[V \setminus D]$ for any $j \in [m]$.

  \textbf{$\Rightarrow$} Let $D$ be a solution of $G$ with $|D| \le k$. Observe first that any solution must pick at least one vertex in each
  set $\{t_{(u,0)},t_{(u,1)}\}$ for $u \in [log(t)]$ and at least one vertex in each set $\{v_{(i,T)},v_{(i,F)}\}$ for $i \in [n]$ (as each of these set is a digon without
  exit arc). As $k=log(t)+n$, we even get that any solution must pick exactly one vertex in each of these sets.
  For any $u \in [log(t)]$, let $b_u$ be the vertex of $D \cap \{t_{(u,0)},t_{(u,1)}\}$ and for any $i \in [n]$ let $a_i$ be the vertex of $D \cap \{v_{(i,T)},v_{(i,F)}\}$.
  Let $A$ be the assignement such that $x_i$ is true if $a_i = v_{(i,F)}$, and false otherwise. Let $\ell_0$ such that $(b^{\ell_0}_1,\dots,b^{\ell_0}_{log(t)})=(b_1,\dots,b_{log(t)})$.
  Let us prove that $A$ satisifies $I^{\ell_0}$. Let $j \in [m]$. Notice that the $log(t)$ endpoints of the out arcs from $c^{\ell_0}_j$ to $T$ belong to $D$.
  This implies that the only possible path from $c^{\ell_0}_j$ to a sink in $G[V \setminus D]$ is just an arc from $c^{\ell_0}_j$ to a vertex $v \in V \setminus D$.
  This implies that the litteral associated to this variable satisfies $C^{\ell_0}_j$ in $A$, and thus that any clause of $I^{\ell_0}$ is satisfied.
\end{proof}

\fi

\section{Distance to DAG having a Hamiltonian path}\label{sec:hamiltonian}
In this section we present a polynomial Turing kernel when we are given a special directed feedback vertex set whose the removal returns an acyclic graph having a directed Hamiltonian path.

\if 10

\begin{defi}
  Let $R$ be a dag.
  Given any antichain $K \subseteq R$ (no path in $R$ between two vertices of $K$) and $v \in R \setminus K$ such that $v$ is accessible from any $s \in K$,
   we define the following oriented graph $G(v,K) = (K,A)$, where $s_1s_2 \in A$ iff $v$ is no longer accessible from $s_1$ in
  $R \setminus N^+(s_2)$ (either because $v \in N^+(s_2)$ or because $N^+(s_2)$ kills all the
  paths from $s_1$ to $v$). 
\end{defi}

\begin{lema}
   Let $R \subseteq V$ such that $G[R]$ is a dag.
   Let us consider an opt solution such that $K \subseteq R$ (implying that $K$ is an antichain according to previous lemma).
   For any $v \in R \setminus K$, let $K_v \subseteq K$ be the sinks that have a path (in $R$) to $v$. For any $v \in R \setminus K$, $G(v,K_v)$ has a universal (in) vertex.
\end{lema}
\begin{proof}
  Suppose by contradiction that this is not the case. Let $v \in R \setminus K$ such that $G(v,K_v)$ has no universal in vertex.
  Notice that this implies that $v \notin S$.
  Let $s$ be a source accessible from $v$ in the fixed optimal. As $K \subseteq R$, we know that $s \in R$.
  If $s \notin K_v$, then take any $s \in K_v$, let $P$ be the vertices of a $sv$-path (including $s$ and $v$), and let $X = P \cap S$.
  Removing $X$ from $S$ is safe as $s$ will remain a source. Indeed, $N^+(s) \cap X = \emptyset$, as otherwise $s \in K_v$.
  Othserwise ($s \in K_v$), as $G_v$ has no universal in vertex, let $s'$ such that $s's \notin A$.
  This immplies that there $v$ is still accessible from $s_1$ in $R \setminus N^+(s_2)$. Let $P$ be the vertices of such a path,
  and $X = P \cap S$. Again, removing $X$ from $S$ is safe as $s$ will remain a source (as $X \cap N^+(s) = \emptyset$).
  
\end{proof}

\begin{cor}
  Given a dag $R$ with a sink $v$, let $\alpha(R)$ be the largest size of an antichain $K$ such that $G(s,K)$ has a universal in vertex.
  If $\alpha(R)$ is constant, then there is a the pb is FPT.
(so for example works in any grid oriented to the right and down)
\end{cor}
proof
first guess what happen in the modulator and propagate.
This may break $R$ into several CC $R_i$ which are vertex induced subgraphs (?) of $R$.
Prove that $\alpha(R_i)$ is also bounded by constant.
Then do a DP which branch and take any subset to guess $K \cap R_i$, and for each guess encodes whan happen in the modulator...

\fi

\begin{lemma}[Single sink along a path]\label{lem:singlesink}
  Let $G$ be a directed graph, and let $R \subseteq V(G)$ such that $G[R]$ is a DAG.
  Let $P$ be any path in $R$.
  Then in a minimum knot-free vertex deletion set $S$ of $G$ with set of sinks $Z$ in $G_{\bar{S}}$, we have $|Z \cap V(P)| \le 1$. 
\end{lemma}
\begin{proof}
Assume by contradiction that $|Z \cap V(P)| \ge 2$. Let $P=(v_1,\dots,v_p)$, and let $i_1$, $i_2$ be the indices of two consecutive vertices of $Z \cap V(P)$, or more formally such that $i_1 \le i_2$, $\{v_{i_1}, v_{i_2}\} \subseteq Z \cap V(P)$, and
for any $i \in ]i_1,i_2[$, $v_i \in V(P) \setminus Z$. Let $P'=(v_{i_1},\dots,v_{i_2})$ be the $v_{i_1}v_{i_2}$ subpath of $P$.

Let $u = N^+(v_{i_1}) \cap V(P)$. Observe that $u \neq v_{i_2}$ (as otherwise $v_{i_2}$ would be in $S$ and not in $Z$) and that $u \in S$. 
Let $S_{P'} = S \cap V(P')$. Observe that $S_{P'} \neq \emptyset$ as it contains $u$. Let $v$ be the last (in the order of $P'$) vertex of $S_{P'}$. Notice that $v \notin N^+(v_{i_2})$ because $P$ is in the DAG $R$. Thus, we get that $v_{i_2}$ is still a sink in $G_{\bar{S'}}$ (where $S' = S \setminus \{v\}$), and by Lemma~\ref{lem:addback} we conclude that $S'$ is still a solution, which is a contradiction.
\end{proof}

\begin{lemma}[Useless vertices]\label{lem:useless}
Let $B$ be a subset of vertices of $G$. Let $G'$ be the graph obtained by applying the following operation for every $v \in B$ (in any order): remove $v$ and add all arcs between $N^-(v)$ to $N^+(v)$ (removing any loop which could appear).
\begin{enumerate}
    \item if there exists an optimal solution $S$ with set of sinks $Z$ in $G_{\bar{S}}$ such that $B \cap (S \cup Z)=\emptyset$, then $S$ is still a solution of $G$.
    \item for any solution $S$ of $G'$, $S$ is still a solution of $G$.
\end{enumerate}
\end{lemma}
\begin{proof}
The proof is by induction on $|B|$. Let us start with $B=\{v\}$.
$1)$. Let us prove first that $S$ is still a solution of $G'$. Let $Z'$ be the set of sinks of $G'_{\bar{S}}$. Notice that $v \notin S$ implies that $N^-(v) \cap Z = \emptyset$. Let $u \in V(G')$. As $S$ is a solution of $G$, let $P$ be a $uz$-path in $G_{\bar{S}}$ with $z \in Z$. Notice first that as $N^-(v) \cap Z = \emptyset$, we know that $z \notin N^-(v)$. As $v \notin Z$, we also get that $z\neq v$, implying that $z \in Z'$. If $v \notin V(P)$, then this path still exists in $G'$ and we are done. 
Otherwise, $P$ contains the subpath  
$v_1vv_2$ with $v_1 \in N^-(v)$ and $v_2 \in N^+(v)$. Replacing this subpath by $v_1v_2$ (which is an arc in $G'$), we get a $uz$-path in $G'$ with $z \in Z'$.

$2)$.
Let us now prove that for any solution $S'$ of $G'$ with set of sinks $Z'$ in $G'_{\bar{S'}}$, $S'$ is still a solution in $G$. Let $\tilde{Z}$ be the set of sinks in $G_{\bar{S'}}$. Let us first consider a vertex $u \neq v$ and prove that there exists a $u\tilde{Z}$-path in $G_{\bar{S'}}$. As $S'$ is a solution of $G'$ there exists a $uz$-path in $G'$ with $z \in  Z'$. First, remove any two consecutive vertices $v_1$, $v_2$ in $P$ such that $v_1 \in N^-(v)$ and $v_2 \in N^+(v)$, and replace them by $v_1vv_2$. This gives a $uz$-path $P'$ in $G_{\bar{S'}}$. If $z \notin N^-(v)$, then $z \in \tilde{Z}$
($z$ remains a sink in $G_{\bar{S'}}$). Otherwise, $z \in N^-(v)$ being a sink in $G'_{\bar{S'}}$ implies that $N^+(v) \subseteq S'$, implying in turn that $v \in \tilde{Z}$. In this case, we add $v$ at the end of $P'$ and we get the desired path. Let us now prove that there exists a $v\tilde{Z}$-path in $G_{\bar{S'}}$.
If  $N^+(v) \subseteq S'$ then $v \in \tilde{Z}$ is a sink in $G_{\bar{S'}}$, otherwise there exists in $G_{\bar{S'}}$ an arc $vu$, and according to the previous case we know that there exists a $u\tilde{Z}$-path in $G_{\bar{S'}}$.

Let us now suppose that $|B|>1$.
Let $v \in B$ and $G'$ the graph defined above obtained after removing $v$. According to item $1$, $S$ is still a solution of $G'$, and according to item $2$, it is still an optimal solution of $G'$. This implies that we can apply induction on $B \setminus \{v\}$.
\end{proof}



\begin{lemma}\label{lemma:kernelpath}
Let $G$ be a directed graph and $F,R$ be a partition of $V(G)$. In polynomial time we can construct a graph $G'$ with $k|F|$ vertices such that
  \begin{itemize}
      \item if $(G',k)$ is a {\em yes}-instance then $(G,k)$ is a {\em yes}-instance;
      \item if $(G,k)$ is a {\em yes}-instance and there exists an optimal solution $S$ with set of sinks $Z$ in $G_{\bar{S}}$ such that $R \cap Z = \emptyset$, then $(G',k)$ is a {\em yes}-instance 
  \end{itemize} 
  \end{lemma}
\begin{proof}
Let us partition $F$ into $F_{\le k}=\{v \in F \mbox{ s.t. } |N^+(v) \cap R| \le k\}$ and $F_{> k}=F\setminus F_{\le k}$.
Let $A  = R \cap N^+(F_{\le k})$ and $B = R \setminus A$. Notice that there may be some arcs from $F_{> k}$ to $A$, but not from $F_{\le k}$ to $B$.
Define $G'$ by applying Lemma~\ref{lem:useless} on set $B$. 
We get that $G'$ has $k|F|$ vertices, and that if $(G',k)$ is a {\em yes}-instance then $(G,k)$ is a {\em yes}-instance. 
Suppose now that $(G,k)$ is a {\em yes}-instance and there exists an optimal solution $S$ with set of sinks $Z$ of $G_{\bar{S}}$ such that $R \cap Z = \emptyset$.
According to Observation~\ref{obs:degree}, $F_{> k} \cap Z = \emptyset$. Thus, as 
$R \cap Z = \emptyset$ and $N^+(Z)=S$, we get that $B \cap S = \emptyset$ and 
Lemma~\ref{lem:useless} gives us the desired property.
\end{proof}

\begin{proposition}
 Suppose $V(G)$ is partitioned into $F$ and $R$ where $R$ is a DAG with a Hamiltonian path. Then, there is a polynomial Turing kernel with $\mathcal{O}(k|F|)$ vertices.
\end{proposition}
\begin{proof}
Let $S$ be an optimal solution with set of sinks $Z$ in $G_{\bar{S}}$
By Lemma~\ref{lem:singlesink}, and since $R$ has an Hamiltonian path $P$, we get that $|Z \cap V(P)|=|Z \cap R| \le 1$. Informally, we will guess (among $|R| \le n$ choices) the potential vertex in $Z \cap R$, move it to $F$, and apply the previous kernel. Thus, for each of these $n$ choices we will obtain a shrinked input of size $\mathcal{O}(|F|k)$ that we can solve using an oracle, and we will answer {\em yes} iff one of these $n$ reduced instances is a {\em yes}-instance. More formally, let $R = \{v_1,\dots,v_r\}$. For any $i$, let $F_i = F \cup \{v_i\}$, $R_i= R \setminus \{v_i\}$ and $G'_i$ the graph obtained by applying Lemma~\ref{lemma:kernelpath} on $F_i$, $R_i$.
As $|G'_i|\le \mathcal{O}(|F|k)$, we can make an oracle call on each $G'_i$ to get an answer $a_i$ and output {\em yes} iff one of the $a_i$ is {\em yes}. If $(G,k)$ is a {\em yes}-instance (and $S$, $Z$ the associated solution), then there exists $i^*$ such that $R_{i^*} \cap Z = \emptyset$, implying that $(G'_{i^*},k)$ will be a {\em yes}-instance and that we will return {\em yes}. On the other side, if a $(G'_i,k)$ is a {\em yes}-instance then by Lemma~\ref{lemma:kernelpath} we get that $(G,k)$ is a {\em yes}-instance.  
\end{proof}

\if 10
\begin{lemma}
If $R$ is an IS, then FPT (already know as fpt / tw but gives a second example of fpt at "the opposite" of having a Hamil path.
\end{lemma}
Do as before with $A$ and $B$ : move $A$ to $X$ to get $X'$ and $R'=B$. As $R$ was an IS, here we know that $R' \cap S = \emptyset$ (and not $K$).
Thus, guess $X' \cap S$ and check if it a solution.

\begin{lemma}
  If $R$ is, for example, an arbitrary large collection $R_i$ where $R_i$ is a dag with a Hamil path, I think this is ftp
\end{lemma}
Guess $X \cap S$, then do a DP wich starts with $R_1$. Branch in $n$ to decide which is the guy in $R_1 \cap K$ (must respect the fixed deleted part in $X$), and mark accordingly
the guys in $X$ that became connected to a sink.

I don't know if we can black box (if dist to bla is FPT, then dist to an arbitrary large collection of bla is FPT)

Ideas for FPT for an arbitrary dag : partition $R$ into $A$ and $B$. Guess what happen in $X \cup A$ and clean.
It remains $X'$ and $R'$ (subdag) where $X' \cap (S \cup K) = \emptyset$.
Now try to use observations on sinks in $R'$. If we can locate one sink of S in FPT, then ok.
Let a color of $x \in R'$ be $(N^+(x) \cap X', |N^+(x) \cap R'|)$ (only card for second part as no hope to have true twins, ex in a path). There is an FPT number of colors.
Intuition that for guys in the same colors and such that there is a path from $v_1$ to $v_2$, if $v_1 \in K$ (and thus $v_2 \notin K$), it SHOULD be better to take $v_2$
as it cost the same. And thus we could simply guess the color of a sinkg and take a minimal element?
\fi


\bibliography{carneiro-et-al}

\end{document}